\documentclass[11pt,letterpaper]{article}
\usepackage[margin=1in]{geometry}
 \usepackage{authblk}
\usepackage{graphicx}
\usepackage{amsmath}
\usepackage[colorlinks=True, allcolors=blue]{hyperref}
\usepackage{array}
\usepackage[linesnumbered,ruled,vlined]{algorithm2e}
\usepackage{xcolor} 

\SetCommentSty{mycommfont}
\usepackage{xcolor}

\SetCommentSty{mycommfont}
 \usepackage{amsthm}
 
\newtheorem{theorem}{Theorem}[section]

\newtheorem{lemma}[theorem]{Lemma}

\title{A Space-Efficient Algorithm for  Longest Common Almost Increasing Subsequence of Two Sequences} 

\author[1]{Md Tanzeem Rahat}
\author[1]{Md. Manzurul Hasan}
\author[2]{Debajyoti Mondal}
\affil[1]{Department of Computer Science, \protect\\ American International University-Bangladesh, Bangladesh \protect\\ \texttt{tanzeem@aiub.edu,manzurul@aiub.edu}} 
\affil[2]{Department of Computer Science, University of Saskatchewan, Saskatoon,  Canada \protect\\ \texttt{dmondal@cs.usask.ca}}

\begin{document}

\maketitle

\begin{abstract}
Let $A$ and  $B$ be two number sequences of length $n$ and $m$, respectively, where $m\le n$. Given a positive number $\delta$, a common almost increasing sequence $s_1\ldots s_k$ is a common subsequence for both $A$ and $B$ such that for all $2\le i\le k$, $s_i+\delta > \max _{1\le j < i} s_j$. The LCaIS problem seeks to find the longest common almost increasing subsequence (LCaIS) of $A$ and $B$. An LCaIS can be computed in $O(nm\ell)$ time and $O(nm)$ space [Ta, Shieh, Lu (TCS 2021)], where $\ell$ is the length of the LCaIS of $A$ and $B$. In this paper we first give  an  $O(nm\ell)$-time and $O(n+m\ell)$-space algorithm to find LCaIS, which improves the space complexity. We then design an $O((n+m)\log n +\mathcal{M}\log \mathcal{M} + \mathcal{C}\ell)$-time and $O(\mathcal{M}(\ell+\log \mathcal{M}))$-space algorithm, which is faster when the number of matching pairs $\mathcal{M}$ and the number of compatible matching pairs $\mathcal{C}$ are in $o(nm/\log m)$. 
\end{abstract}

\section{Introduction}

The \emph{longest increasing subsequence} (LIS) problem is a classical problem in computer science  that has been extensively studied over the past few decades. The problem seeks for the longest subsequence of a given sequence in which the elements are arranged in strictly increasing order. 
Schensted \cite{schensted1961longest} employed a dynamic programming approach that efficiently computes an LIS in $O(n \log k)$ time, where $n$ is the length of the input sequence  and $k$ is the length of the longest increasing subsequence.  
A natural extension of the LIS problem is the \emph{longest common increasing subsequence} (LCIS) problem, where given two number sequences, one needs to find a longest increasing subsequence that must also be  common to both the given  sequences. The problem of finding a longest common subsequence (LCS) of two sequences and its variants appear in many  practical applications such as matching DNA sequences, analyzing stock market, and so on.  
Although there exists a quadratic-time algorithm for LCIS~\cite{yang2005fast}, an $O(n^{2-\epsilon})$-time algorithm, where $\epsilon>0$, cannot exist unless the Strong Exponential Time Hypothesis (SETH) fails \cite{abboud2015tight,bringman2018multivariate}.  
Consequently, research has focused on developing algorithms with output-dependent running times. Kutz et al. \cite{kutz2011faster} proposed an algorithm to compute LCIS in $O((n + m \cdot l) \log \log \sigma + \text{Sort})$ time, where $n$ and $m$ are the lengths of the input sequences, $l$ is the length of the LCIS, $\sigma$ is the size of the alphabet, and Sort is the time required to sort the sequences. More recently, subquadratic-time algorithms have been developed with running times $O((n\log \log n)^2/{\log^{1/6} n})$~\cite{duraj2020sub} and $O(n^2\log \log n/\sqrt{\log n})$~\cite{agrawal2020faster}, respectively.

The \emph{longest almost increasing sequence} (LaIS) problem, introduced by Elmasry \cite{elmasry2010longest}, is a relaxed version of the LIS problem, where the goal is to find a longest subsequence of a given sequence that is ``almost'' increasing. Specifically, for a sequence $A=(s_1, s_2, \ldots, s_n)$ and a constant $\delta > 0$, the LaIS problem seeks to construct a longest possible subsequence $(s_1, s_2, \ldots, s_l)$ of $A$ such that for all $2\le i \le k$, $ s_i + \delta > \max _{1 \leq j < i} s_j$. If we allow $\delta$ to be 0, then the LaIS problem becomes equivalent to the LIS problem. The LaIS problem helps to deviate from strict monotonicity, making it applicable in scenarios where noise or fluctuations are present in the data. An interesting example is stock prices that might exhibit an almost increasing (or decreasing) trend without being strictly increasing. Elmasry \cite{elmasry2010longest} provided an algorithm that solves the LaIS problem in $O(n \log l)$ time, which is asymptotically optimal for the comparison-tree model.

The \emph{longest common almost increasing subsequence} (LCaIS) problem \cite{moosa2013computing,ta2021computing}, combines the challenges of the LaIS and LCIS problems. The LCaIS problem takes two number sequences $A $ and $B $, along with a constant $\delta > 0$, and seeks to find the longest common subsequence   of $A$ and $B$ that is almost increasing. Formally, the subsequence  $(s_1, s_2, \ldots, s_l)$  must satisfy the condition $ s_i + \delta > \max _{1 \leq j < i} s_j$ for all $2\le i \le k$.   Moosa, Rahman and Zohora \cite{moosa2013computing} claimed an $O(n^2)$-space and $O(n(n + \delta^2))$-time algorithm for computing LCaIS for two permutations, each of length $n$, which they also claimed to be applicable for inputs that do not contain repeated elements. 
Ta, Shieh, and Lu  \cite{ta2021computing} identified a flaw in this algorithm and  presented a dynamic programming algorithm that can compute an LCaIS  in $O(nm\ell)$ time and $O(nm)$ space, where $\ell$ is the length of the LCaIS.  
Bhuiyan, Alam, and  Rahman~\cite{bhuiyan2022computing} gave several approaches to compute LCaIS, but all take $O(nm)$ space if an LCaIS needs to be computed. They also gave an $O(nm+\mathcal{M}(p+\ell))$-time and $O((\mathcal{M}+m)\delta)$-space algorithm, where $\mathcal{M}$ is the number of matching pairs and $p$ is the longest common subsequence of the input sequences. This yields faster running time compared to~\cite{ta2021computing} in some specific cases, e.g., when  input sequences do not contain any  repeated element.   


\smallskip
\noindent\textbf{Contribution.} In this paper we show how to compute an LCaIS   in $O(nm\ell)$ time and $O(n+m\ell)$ space, where $\ell$ is the length of the LCaIS. Note that all previously known algorithms take $O(nm)$ space~\cite{ta2021computing,bhuiyan2022computing}. We adapt a seminal work of Hirschberg~\cite{hirschberg1975linear} to obtain this improved space complexity. We then design an 
$O((n+m)\log n +\mathcal{M}\log \mathcal{M} + \mathcal{C}\ell)$-time and $O(\mathcal{M}(\ell+\log \mathcal{M}))$-space
algorithm, which is faster when $\mathcal{M}$ and $\mathcal{C}$ are in $o(nm/\log m)$. Here $\mathcal{M}$ is the number of matching pairs between the two input sequences and   $\mathcal{C}$ is the number of compatible matching pairs, where two matching pairs are called \emph{compatible} if they form a common subsequence in the given input sequences. To design this algorithm, we transform the LCaIS problem into the problem of finding a special path in a directed graph. The vertices of the graph correspond to the matching pairs, and the edges are formed between compatible matching pairs that also satisfy the almost increasing condition. We then show an efficient way to find a path that corresponds to an LCaIS.  This algorithm appears to be interesting when we observe that all previously known algorithms take $\Omega(nm)$ time~\cite{ta2021computing,bhuiyan2022computing} and there are cases when $\mathcal{C}$ can be asymptotically smaller than $\mathcal{M}$. 


The rest of the paper is structured as follows. Section~\ref{sec:over} gives an overview of Hirschberg's Algorithm~\cite{hirschberg1975linear} for computing LCS in $O(n+m)$ space.
 Section~\ref{sec:ours} presents our proposed algorithm for computing LCaIS.  Section~\ref{sec:non} discusses the case when we consider compatible matching pairs. Finally, Section~\ref{sec:con} concludes the paper.  

\section{Overview of Hirschberg's Algorithm~\cite{hirschberg1975linear} for LCS}
\label{sec:over}
Let $A[1\ldots n]$ and $B[1\ldots m]$ be two number sequences of length $n$ and $m$, respectively. By $\mathcal{L}(i,j)$ we denote the LCS length of $A[1\ldots i]$ and $B[1\ldots j]$, where $0\le i\le n$ and $0\le j\le m$. By $\mathcal{L}'(i,j)$ we denote the LCS length of $A[i+1\ldots n]$ and $B[j+1\ldots m]$, where $1\le i+1\le n$ and $1\le j+1\le m$. Given two number sequences $S_1$ and $S_2$, we denote their concatenation as $S_1 \circ S_2$.

Hirschberg~\cite{hirschberg1975linear} gave an $O(nm)$-time and $O(n+m)$-space algorithm for the LCS problem.  The idea is to first design an $O(nm)$-time  linear-space algorithm (Algorithm~\ref{algo:sub}) that computes $L[j]$, i.e., the LCS length for every $A[1\ldots n]$ and $B[1\ldots j]$, where $1\le j\le m$. An LCS is then computed using a divide-and-conquer approach (Algorithm~\ref{algo:dc}) within the same time and space, where  Algorithm~\ref{algo:dc} is used to decompose the problem into smaller subproblems.

Algorithm~\ref{algo:sub} iterates over the sequence $A$ and at the start of an iteration $i$, it assumes that the LCS length  $\mathcal{L}(i-1,j)$ is already computed for $0 \le j\le m$ and stored in an array $prev[1\ldots m]$.  In the $i$th iteration, it computes $L[1\ldots m]$ by iterating over $j$ from $1$ to $m$. If $A[i]\not= B[j]$, then $L[j]$ is assigned $\max\{\mathcal{L}(i,j-1), \mathcal{L}(i-1,j)\}$. Here $\mathcal{L}(i,j-1)$ is already computed at the $(j-1)$th iteration and $\mathcal{L}(i-1,j)$ can be accessed from $prev[j]$. If $A[i] = B[j]$, then $L[j]$ is the same as  $\mathcal{L}(i-1,j-1) +1$, i.e., $prev[j-1]+1$. Note here that $\mathcal{L}(i-1,j)$ and $\mathcal{L}(i,j-1)$ cannot be larger than $\mathcal{L}(i,j)$.

\begin{algorithm}[h]
\caption{ComputeL($A[1\ldots n], B[1\ldots m], L[1\ldots m]$)}
\label{algo:sub} 
\SetAlgoLined
   $L[0\ldots m] \gets [0, \ldots, 0]$\; 
    \For{$i \gets 1$ \KwTo $n$}{
           $prev[0 \ldots m] \gets L[0, \ldots, m]$\;  

        \For{$j \gets 1$ \KwTo $m$}{
            \lIf{$A[i] = B[j]$}{
                $L[j] \gets prev[j-1]+1$
            }
            \lElse{
                $L[j] \gets \max\{ L[j-1], prev[j]\}$
            }
        }
    } 
\end{algorithm}

The divide-and-conquer approach for constructing the LCS is based on the observation that for every $0\le i\le n$ there exists some $j$, where $0\le j\le m$,  such that the LCS length is equal to   
$M(j)=\mathcal{L}(i,j) + \mathcal{L}'(i,j)$. Therefore, it suffices to find the index $j$ that maximizes $M(j)$ and divide the problem of constructing LCS into two subproblems. It may initially appear that the above equation should use $\mathcal{L}'(i+1,j+1)$ instead of $\mathcal{L}'(i,j)$, but note that $\mathcal{L}'(i,j)$ denotes the LCS length of $A[i+1\ldots n]$ and $B[j+1\ldots m]$. 

\begin{algorithm}[h]
\caption{LCS($A[1\ldots n], B[1\ldots m],S$)}
\label{algo:dc} 
   \lIf {$m=0$}{$S\gets $ empty sequence, return}
   \If {$n=1$}{
        \lIf {$\exists j$ where $A[1]=B[j]$}{$S\gets A[1]$}
        \lElse{$S\gets $ empty sequence}
        return\;
   } 
   $i\gets \lfloor n/2\rfloor$\;
   ComputeL($A[1\ldots i], B[1\ldots m], L[1\ldots m]$)\;
   ComputeL($A[i+1\ldots n], B[1\ldots m], L'[1\ldots m]$)\;
   Determine an index $j$ that maximizes $L[j]+L'[j]$\tcp*{Problem decomposition}
   LCS($A[1\ldots i], B[1\ldots j], S_1$)\;
   LCS($A[i+1\ldots n], B[j+1\ldots m], S_2$)\;
   $S\gets S_1 \circ S_2$ 
\end{algorithm}

\section{Finding LCaIS in $O(nm\ell)$ Time and $O(n+m\ell)$ Space }
\label{sec:ours}
We now show that the divide-and-conquer framework of Hirschberg~\cite{hirschberg1975linear} can be adapted to compute LCaIS in $O(nm\ell)$ time and $O(n+m\ell)$ space, where $\ell$ is the length of the LCaIS for input sequences $A[1\ldots n]$ and $B[1\ldots m]$.

The idea is to  design two subroutines (Section~\ref{sec:pq}) --- one that computes the LCaIS for $A[1 \ldots  n]$ and $B[1 \ldots j]$ ending  at $B[j]$ and the other that computes the LCaIS for $A[1 \ldots  n]$ and $B[j \ldots m]$ starting at $B[j]$, where $1\le j \le m$. We also ensure that these LCaIS possess additional properties. The former should minimize the largest element used in the LCaIS, while the latter should maximize the smallest element.  We use the former subroutine to find LCaIS for the sequences $A[1\ldots \lfloor n/2\rfloor ]$ and $B[1\ldots m]$, and the latter subroutine to compute LCaIS for the sequences $A[\lfloor n/2\rfloor +1\ldots n]$ and $B[1\ldots m]$. We use the solutions of these subroutines to divide the problem of computing an LCaIS into solving  two subproblems (Section~\ref{sec:divcon}). Specifically, we find a suitable index $j$ that ensures that a common increasing subsequence ending at $B[j]$ can be concatenated with a common increasing subsequence starting at $B[j]$ to form the desired LCaIS. 

\subsection{Deciding LCaIS length when Minimizing   Largest Element}
\label{sec:pq}
We first design an $O(nm\ell)$-time and $O(m\ell)$-space subroutine similar to Algorithm~\ref{algo:sub} that computes $P'[j,r]$, where $1\le j\le m$ and $1\le r \le \ell$, and $P'[j,r]$ is True if and only if there exists an LCaIS for $A[1\ldots n]$ and $B[1\ldots j]$ of length $r$ ending with $B[j]$. We will also compute $Q'[j,r]$, i.e., the smallest number such that there exists such an LCaIS for $A[1\ldots n]$ and $B[1\ldots j]$ where the largest element is $Q'[j,r]$. If $P'[j,r]$ is False, then  $Q'[j,r]$ is set to $\infty$.

To compute $P'[j,r]$ and $Q'[j,r]$, we will use the following lemma.

\begin{lemma}\label{lem:pq}
Let $\mathcal{P}'(i,j,r)$ be True if there exists an  LCaIS for $A[1\ldots i]$ and $B[1\ldots m]$ of length $r$ ending with $B[j]$. For $r>0$, if $\mathcal{P}'(i,j,r)$ is True, then let $\mathcal{Q}'(i,j,r)$ denote the smallest number such that there exists an LCaIS ending with $B[j]$ of length $r$ such that no element of the LCaIS is larger than $\mathcal{Q}'(i,j,r)$. Then the following equalities hold. 
\begin{equation*}
\mathcal{P}'(i,j,r ) =\begin{cases}		 
		     \bigvee\limits_{\substack{1\le k< j \\ A[i]+\delta > \mathcal{Q}'(i-1,k,r-1)}}\mathcal{P'}(i-1,k,r-1), & \text{if $A[i]=B[j]$}\\
            \mathcal{P}'(i-1,j,r ) , & \text{otherwise.}
		 \end{cases}
\end{equation*} 


\begin{equation*}
\mathcal{Q}'(i,j,r) =\begin{cases}		 
		     \min\limits_{\substack{1\le k< j \\ A[i]+\delta > \mathcal{Q}'(i-1,k,r-1) \\ \mathcal{P}'(i-1,k,r-1)= True}} \max\left(A[i],\mathcal{Q'}(i-1,k,r-1)\right), & \text{if $A[i]=B[j]$}\\
            \mathcal{Q}'(i-1,j,r ), & \text{otherwise.}
		 \end{cases}
\end{equation*} 
 
\end{lemma}
\begin{proof}
Consider first the case of computing  $\mathcal{P}'(i,j,r)$. 

If $A[i]\not=B[j]$, then LCaIS cannot match $A[i]$ and $B[j]$. Hence  there exists an LCaIS of length $r$ ending at $B[j]$, i.e.,   $\mathcal{P}'(i,j,r)$ is True,  if and only if there exists an LCaIS of length $r$ (ending at $B[j]$) for the sequences $A[1\ldots i-1]$ and $B[1\ldots j]$, i.e., if  $\mathcal{P}'(i-1,j,r)$ is True. 

If $A[i]=B[j]$, then we assume that $A[i]$ has been matched to $B[j]$ in the LCaIS. If there exists an LCaIS of length $r$, then there must be an LCaIS  $\sigma_k$ of length $r-1$ for $A[1\ldots i-1]$  and $B[1\ldots k]$ for some $k$ where $1\le k<j$ such that $A[i]$ can be appended to $\sigma_k$ to form an LCaIS of length $r$. However, to append $A[i]$ to $\sigma_k$, $\sigma_k$ must be compatible with $A[i]$, i.e., the largest element of $\sigma_k$ must be smaller than $A[i]+\delta$. Hence we can restrict our attention to those indices where $A[i]+\delta> \mathcal{Q}'(i-1,k,r-1)$. Since $\mathcal{Q}'(i-1,k,r-1)$ minimizes the largest element over all $\sigma_k$ of length $r-1$, if there exists some compatible $\sigma_k$, then we must have $A[i]+\delta> \mathcal{Q}'(i-1,k,r-1)$ where $\mathcal{P}'(i-1,k,r-1)$ is True.
 
Consider now the case of computing  $\mathcal{Q}'(i,j,r)$. 

If $A[i]\not=B[j]$, then LCaIS cannot match $A[i]$ and $B[j]$.  Therefore, if there exists an LCaIS of length $r$ ending with $B[j]$, then that would also be an LCaIS of length $r$ (ending with $B[j]$) for the sequences $A[1\ldots i-1]$ and $B[1\ldots j]$.   Since $\mathcal{Q}'(i-1,j,r)$ minimizes the largest element over all LCaIS of length $r$ that ends with $B[j]$, $\mathcal{Q}'(i,j,r)$ can be obtained from $\mathcal{Q}'(i-1,j,r)$.

If $A[i]=B[j]$, then similar to the computation of $\mathcal{P}'(i,j,r)$, we can restrict our attention to the candidate indices $k$ where $1\le k<j$ and $A[i]+\delta > \mathcal{Q}'(i-1,k,r-1)$. In addition, we can only consider the cases when $\mathcal{P}'(i-1,k,r-1)$ is True because we want the final LCaIS to be of length $r$. Consider one such LCaIS  $\sigma_k\circ A[i]$. Since $B[j] (=A[i])$ is the last element of the LCaIS, we find the largest element by taking the maximum of $A[i]$ and $\mathcal{Q}'(i-1,k,r-1)$. Finally, we take the minimum over all candidate indices $k$. 
\end{proof}

We now design a subroutine \textsc{MinimizeLargestElem} based on Lemma~\ref{lem:pq} that computes $P'[j,\ell']$ and $Q'[j,\ell']$ for every $j$ from $1$ to $m$ and for every $\ell'$ from 1 to $\ell$. 

\paragraph{Implementation Details of \textsc{MinimizeLargestElem}.} We iterate over the sequence $A$ and at the start of an iteration $i$, we assume that $P'[j,0 \ldots \ell']$ is already computed for $A[1\ldots i-1]$ and $B[1\ldots m]$ where $0 \le j\le m$, and $\ell'$ is the length of the LCaIS found so far.  In the implementation of $P'$, for each index $j$, we maintain a list of tuples $\mathcal{L}_j$. Let $(r,a_r,b_r)$ be the $r$th tuple of  $\mathcal{L}_j$. Then $a_r=\mathcal{P}'(i-1,j,r)$ indicates whether there exists an LCaIS of length $r$ (ending at $B[j]$), and if so, then $b_r$ is equal to $\mathcal{Q}'(i-1,j,r)$. In the $i$th iteration, we compute $P'[j,0 \ldots \ell'+1]$ for $A[1\ldots i]$ and $B[1\ldots m]$ by iterating $j$ from $1$ to $m$. 

To compute $P'[j,0 \ldots \ell'+1]$ for  $A[1\ldots i]$ and $B[1\ldots j]$, we iterate over $\mathcal{L}_1\ldots \mathcal{L}_{j-1}$ and maintain a list of tuples $\mathcal{T}$ where the $r$th tuple holds the information about whether an LCaIS of length $r$ has been seen at some $\mathcal{L}_k$, where $k<j$ and  $A[i]+\delta> \mathcal{Q}'(i-1,k,r-1)$. If so, and if $A[i]=B[j]$ then we update the tuples $\mathcal{L}_j$ based on the recurrence relation of Lemma~\ref{lem:pq}. If we find an LCaIS of length $\ell'+1$, then we append a new tuple to $\mathcal{L}_j$ and update $\ell'$, i.e., the length of the LCaIS found so far. Algorithm~\ref{algo:pq} gives a pseudocode for \textsc{MinimizeLargestElem}.  The variable $L[j]$ implements $\mathcal{L}_j$, and the arrays \emph{current} and \emph{curMinElem} implement the role of $\mathcal{T}$. Table~\ref{tab:ex} gives an example computation for two specific sequences.

We next compute $X'[j,r]$, $1\le j\le m$, that indicates whether there exists an LCaIS of length $r$ for $A[1\ldots n]$ and $B[1\ldots j]$ (without specifying the ending element). Similarly, we compute $Y'[j,r]$, $1\le j\le m$, which minimizes the maximum element of the LCaIS.  It is straightforward to compute $X'[j,r]$ and $Y'[j,r]$, where $1\le j\le m$, in $O(m\ell)$ time by iterating over the tuples of each list $\mathcal{L}_j$. 
 
\begin{algorithm}[pt]
\caption{\textsc{MinimizeLargestElem}($A[1\ldots n], B[1\ldots m], \delta, X', Y'$)}
\label{algo:pq} 
    Initialize $L[0\ldots m]$, where $L[i]$ with $0\le i\le m$ is a doubly linked list of tuples.\; Set $L[i]$, where $0\le i\le m$, to $(0,0,\infty)$ \tcp*{Maintain $count(L[j])$ in a variable to retrieve the number of tuples in $O(1)$ time }
    $currentMax \gets 0$  \tcp*{LCaIS length found so far }
    \For{$i \gets 1$ \KwTo $n$}{
        $attemptNewMax\gets currentMax+1$\;
        \For{$z \gets 0$ \KwTo $attemptNewMax$}{  
            $current[z] \gets 0$ \tcp*{Initialize to track current to process $A[i]$}
            $curMinElem[z] \gets \infty$  
        }   
        \For{$j \gets 1$ \KwTo $m$}{
        
            \For{$z \gets 0$ \KwTo $attemptNewMax$\tcp*{Iterate over $L[j]$ simultaneously to retrieve $L[j][z]$ or $L[j][z+1]$   in $O(1)$ time}}{ 
                \If {$z \ge count(L[j])$}{
                 $L[j].append(z, 0, \infty)$ \;
                        $r,p,q  \gets z,0,\infty$}
                \lElse{  $r,p,q \gets L[j][z]$  }          
            }
            \If{$A[i] = B[j]$}{
                \If{$r= 0$}{
                    \If {$count(L[j])\le 1$}{
                        $L[j].append(1, 1, A[i])$\tcp*{No past choices to consider}
                        $currentMax \gets \max(currentMax,1)$
                    }
                    \Else{
                        $L[j][r + 1] \gets  (1, 1, A[i])$ \tcp*{LCaIS length 1; take  $A[i]$}
                        }
                }
                \Else{
                    \If{$current[r]=1$\tcp*{Length-$r$ LCaIS exists; build $r+1$}}{
                        \If {$count(L[j])\le r+1$}{
                            $L[j].append(r+1, 1, \max(A[i],curMinElem[r]))$\; 
                            $currentMax \gets \max(currentMax, r + 1)$
                            }
                        \Else{
                            $a,b,d \gets L[j][r+1]$\;
                            $L[j][r+1] \gets (r + 1, 1, min(d,\max(A[i], curMinElem[r])))$ 
                        }

                    }
                }
                \tcc{Update curMinElem if LCaIS of length $r$ exists}
                $current[r] \gets (current[r] \vee  p)$\; 
                \If{$p=1$}
                    {$curMinElem[r] \gets \min(curMinElem[r],q)$}                
            }
            \Else{

                        \If{$(A[i]+\delta > q)$}{
                            $current[r] \gets (current[r] \vee  p)$\; 
                            \If{$p=1$}{
                                $curMinElem[r] \gets \min(curMinElem[r],q)$
                                }
                        }
            }
        }
    }
    Compute $X'$ and $Y'$ from $L$
\end{algorithm}

\begin{table}[h]
\caption{Computation of $L$ with $    A = [3, 1, 4, 5, 2, 4, 5, 1]$, $
    B = [4, 2, 3, 1, 2, 5, 3, 1]$ and $\delta = 3$.}
    \label{tab:ex}
    \centering
    \resizebox{\linewidth}{!}{
    \begin{tabular}{|p{1cm}| *{5}{>{\centering\arraybackslash}p{2cm}|} }
    \hline
          & $r=0$ & $r=1$ & $r=2$ & $r=3$ & $r=4$ \\ \hline
       $L[1]$ & (0, 0, $\infty$) & \textbf{(1, 1, 4)} & (2, 0, $\infty$) & (3, 0, $\infty$) & (4, 0, $\infty$) \\ \hline
        $L[2]$ & (0, 0, $\infty$) & (1, 1, 2) & \textbf{(2, 1, 4)} & (3, 0, $\infty$) & (4, 0, $\infty$) \\ \hline
        $L[3]$ & (0, 0, $\infty$) & \textbf{(1, 1, 3)} & (2, 0, $\infty$) & (3, 0, $\infty$) & (4, 0, $\infty$) \\ \hline
        $L[4]$ & (0, 0, $\infty$) & (1, 1, 1) & \textbf{(2, 1, 2)} & (3, 0, $\infty$) & (4, 0, $\infty$) \\ \hline
        $L[5]$ & (0, 0, $\infty$) & (1, 1, 2) & (2, 1, 2) & \textbf{(3, 1, 3)} & (4, 0, $\infty$) \\ \hline
        $L[6]$ & (0, 0, $\infty$) & (1, 1, 5) & (2, 1, 5) & (3, 1, 5) & \textbf{(4, 1, 5)} \\ \hline
        $L[7]$ & (0, 0, $\infty$) & \textbf{(1, 1, 3)} & (2, 0, $\infty$) & (3, 0, $\infty$) & (4, 0, $\infty$) \\ \hline
        $L[8]$ & (0, 0, $\infty$) & (1, 1, 1) & (2, 1, 1) & (3, 1, 2) & \textbf{(4, 1, 3)} \\ \hline
    \end{tabular}
    }
\end{table}

Let $\mathcal{P}''(i,j,r)$ be True if there exists an  LCaIS for $A[i+1\ldots n]$ and $B[1\ldots m]$ of length $r$ that starts at $B[j]$. For $r>0$, if $\mathcal{P}''(i,j,r)$ is True, then let $\mathcal{Q}''(i,j,r)$ be the largest number such that there exists an LCaIS (with first element $B[j]$) of length $r$ such that no element of the LCaIS is smaller than $\mathcal{Q}''(i,j,r)$. 
One can design a subroutine \textsc{MaximizeSmallestElem} symmetric  to \textsc{MinimizeLargestElem} for computing 
 $X''[j,r]$, $1\le j\le m$, that indicates whether there exists an LCaIS of length $r$ for $A[1\ldots n]$ and $B[j+1\ldots m]$ (without specifying the starting element) and  $Y''[j,r]$ that maximizes the smallest element. 


\subsection{Divide and Conquer}
\label{sec:divcon}

We now design a divide-and-conquer method to construct the  LCaIS, where we will use \textsc{MinimizeLargestElem} and \textsc{MaximizeSmallestElem} to divide the problem into smaller subproblems. 



We will use the following lemma.

\begin{lemma}\label{lem:dc}
Let $i$ be an index where $0\le i \le n$. Assume that $X',Y'$ have been computed for the sequences $A[1\ldots i]$ and $B[1\ldots m]$, and  $X'',Y''$ have been computed for  $A[i+1\ldots n]$ and $B[1\ldots m]$. Then the LCaIS length of $A[1\ldots n]$ and $B[1\ldots m]$ is equal to $\max\limits_{\substack{0\le j\le m \\ X'[j,r'] = X''[j,r''] = True \\Y''[j,r'']+\delta>Y'[j,r']}} (r'+r'')$. 
\end{lemma}
\begin{proof}
Let $S$ be an LCaIS of $A[1\ldots n]$ and $B[1\ldots m]$. By $|S|$, we denote the length of $S$. We can write $S$ as $S'\circ S''$, where $0\le j\le m$, $S'$ is the LCaIS of $A[1\ldots i]$ and $B[1\ldots j]$, and $S''$ is the LCaIS of $A[i+1\ldots n]$ and $B[j+1\ldots m]$. 
Hence $X'[j,|S'|]$ and $X''[j,|S''|]$ must be True.  Let $u$ and $w$ be the largest element in $S'$ and smallest element in $S''$, respectively.  Then $u\ge Y'[j,|S'|]$ and  $Y''[j,|S''|]\ge w$. Since $w+\delta > u$, we must have $Y''[j,|S''|]+\delta>Y'[j,|S'|]$. Therefore, we must have an index $j$ with  $r'=|S'|$ and $r''=|S''|$ such that $|S| = r'+r''$. However, we still need to prove that the maximum value of $(r'+r'')$ over all options is equal to $|S|$. 


Assume for a contradiction that there exists some $j$ such that $X'[j,r']=X''[j,r''] = True$, $Y''[j,r'']+\delta>Y'[j,r']$, and $(r'+r'')>|S|$. 
Note that $Y'[j,r']$  minimizes the largest element, whereas $Y''[j,r'']$  maximizes the smallest element. Therefore, for an element $u$ corresponding to the sequence $X'[j,r']$
 and for an element $w$ corresponding to the 
 sequence $X''[j,r'']$, we have $w+\delta \ge Y''[j,r'']+\delta>Y'[j,r']\ge u$. 
In this case, the LCaIS of $A[1\ldots i]$ and $B[1\ldots j]$ and the LCaIS of $A[i+1\ldots n]$ and $B[j+1\ldots m]$ can be concatenated to obtain an LCaIS of $A[1\ldots n]$ and $B[1\ldots m]$ with a larger length than $|S|$, which contradicts the optimality of $S$. 
\end{proof}

We now can determine an index $j$ based on Lemma~\ref{lem:dc} to decompose the problem into two smaller subproblems. The pseudocode for the divide-and-conquer approach is given in Algorithm~\ref{algo:dc2}, where $S$ corresponds to the LCaIS.
\begin{algorithm}[h]
\caption{LCaIS($A[1\ldots n], B[1\ldots m], \delta, S$)}
\label{algo:dc2} 
   \lIf {$m=0$}{$S\gets $ empty sequence, return}
   \If {$n=1$}{
        \lIf {$\exists j$ where $A[1]=B[j]$}{$S\gets A[1]$}
        \lElse{$S\gets $ empty sequence}
        return\;
   } 
   $i\gets \lfloor n/2\rfloor$\;
   \textsc{MinimizeLargestElem}($A[1\ldots i], B[1\ldots m], \delta, X' ,Y' $)\;
   \textsc{MaximizeSmallestElem}($A[i+1\ldots n], B[1\ldots m], \delta, X'' ,Y'' $)\;
   Determine an index $j$  that maximizes $(r'+r'')$   satisfying the conditions $0\le j\le m$, $X'[j,r'] = X''[j,r''] = True$ and $Y''[j,r'']+\delta>Y'[j,r']$\;
      Remove the elements that are larger than $Y'[j]$ and smaller than $Y''[j]$ from $B[1\ldots j]$ and $B[j+1 \ldots m]$, respectively\;
   LCaIS($A[1\ldots i], B[1\ldots j], \delta, S_1$)\;
   LCaIS($A[i+1\ldots n], B[j+1\ldots m], \delta, S_2$)\;
   $S\gets S_1 \circ S_2$
\end{algorithm}

\subsection{Running Time and Space Complexity}

Let $T(n,m)$ be the running time and let $\ell$ be the LCaIS length for $A[1\ldots n]$ and $B[1\ldots n]$. Lines 7 and 8 of Algorithm~\ref{algo:dc2} take  $O(nm\ell)$ time.  Line 9 can be computed in $O(mr'r'')$ time, where $r'$ and $r''$ are the lengths of the longest subsequences returned by the two subroutines, and hence in $O(nm\ell)$ time. 
Line 10 takes $O(m)$ time. The subproblems on Lines 11 and 12 are of size $T(n/2,j)$ and $T(n/2,m-j)$, respectively. Assume that there exists a constant  $d$ such that $T(n,m)\le d \cdot nm\ell$ for sufficiently large $n$ and $m$.  
Therefore,
\begin{align*}
    T(n,m) &= T(n/2,j) + T(n/2,m-j)+ O(nm\ell) \\
    &\le  d nj\ell/2 + d n(m-j)\ell/2 + O(nm\ell) \\
        &\le   d nm\ell/2 + d'nm\ell, &\text{ where $d'$ is a constant} \\
        & \le d nm\ell/2 + (d/2)nm\ell, &\text{ by setting $d$ to be $2d'$ }\\
        & \le d nm\ell, &\text{ where $d=2d'$ and $m,n\ge 1$}.         
\end{align*}

The space complexity of Algorithm~\ref{algo:pq} is dominated by the lists that are used to implement $P'$ and $Q'$. Since there are $O(m)$ lists, each of size $O(\ell)$, the space complexity becomes $O(m\ell)$. However, we also need to consider the space complexity for Algorithm~\ref{algo:dc2}. Lines 7--10 use a temporary space of size $O(m\ell)$. Therefore, it suffices to consider the space required for the recursive calls. The number of recursive calls, i.e., $R(n,m)$ can be written as $R(n/2,j)+R(n/2,m-j)+1$. Assume that there exists a constant  $d$ such that $R(n,m)\le dn -1$ for sufficiently large $n$. Then 
\begin{align*}
    R(n,m) &= R(n/2,j) + R(n/2,m-j)+ 1 \\
    &\le  d n/2  + d n/2 - 1 \\
        &\le   d n -1, &\text{ where $d=2$ and $n\ge 1$}.     
\end{align*}
We thus obtain the following theorem.
\begin{theorem}
    Given two number sequences $A[1\ldots n]$ and $B[1\ldots m]$ where $m\le n$, and a positive constant $\delta$, an LCaIS for these sequences can be computed in $O(nm\ell)$ time and $O(n+m\ell)$ space, where  $\ell$ is the length of the LCaIS.
\end{theorem}

\section{LCaIS with Few Compatible Matching Pairs  }
\label{sec:non}
In this section we show that the LCaIS can be computed faster for the case when the number of matching pairs $\mathcal{M}$ and the number of compatible matching pairs $\mathcal{C}$ are in $o(nm/\log m)$. For each $B[k]$ we precompute the set of  all indices $match(k)$ such that $A[i]=B[k]$, where $i\in match(k)$. We can do this in $O((n+m)\log n+\mathcal{M})$ time by first sorting the numbers in $A$ and then searching the elements of $B$ using binary search. 
Figure~\ref{fig:ex}(a) shows an input and their matching pairs.

\begin{figure}[h]
    \centering
    \includegraphics[width=\linewidth]{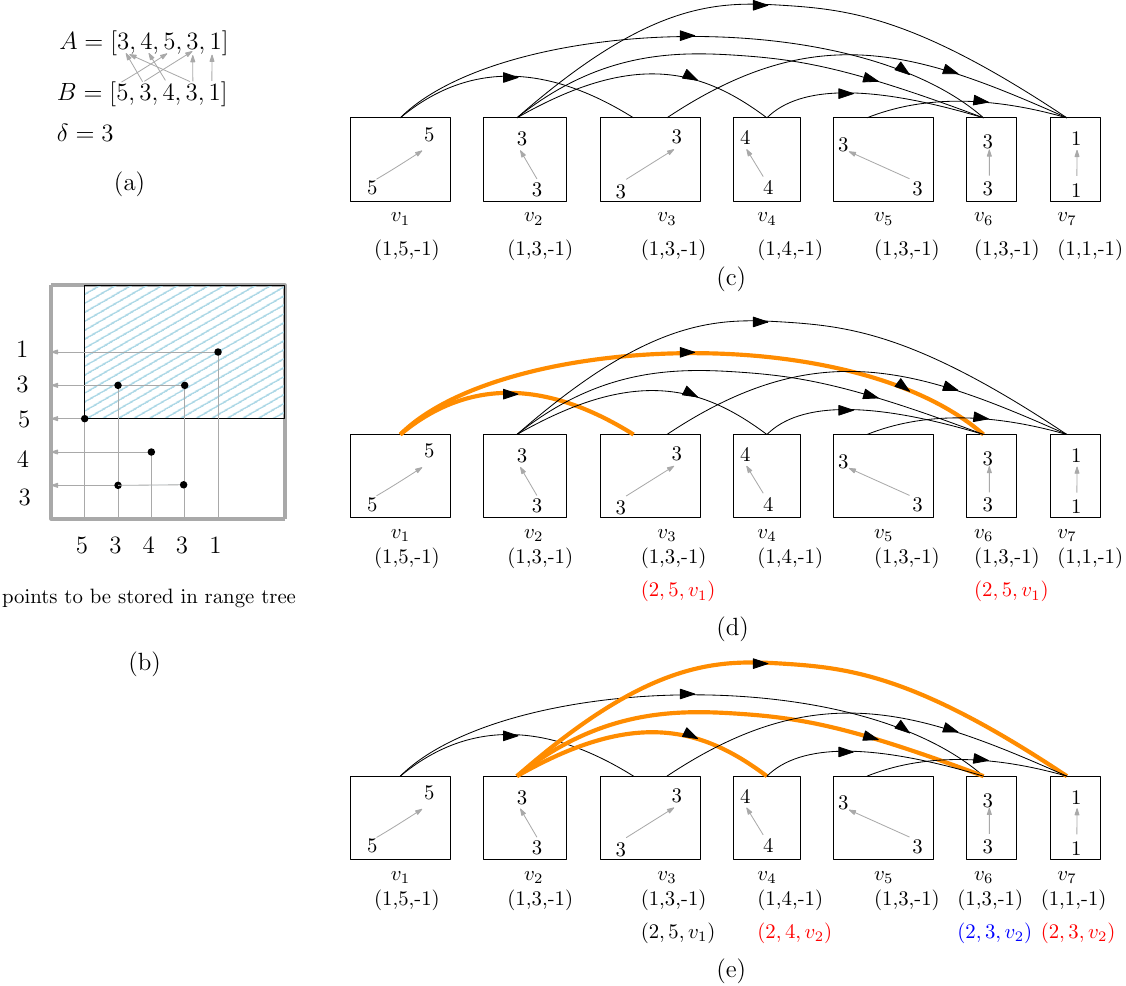}
    \caption{Illustration for the computation of LCaIS length.}
    \label{fig:ex}
\end{figure}
The idea now is to form a directed acyclic graph $G$ and formulate the LCaIS problem as finding a path in the graph. The graph $G$ consists of $\mathcal{M}$ nodes $v_1,\ldots,v_{\mathcal{M}}$ corresponding to the matching  pairs where vertices are ordered so that they first correspond to all the matching pairs for $B[1]$, and then all the matching pairs for $B[2]$, and so on.  

Let $\mathcal{M}_i$ be the matching pair corresponding to $v_i$ and let $e_i$ be the number corresponding to $v_i$. 
The graph $G$ has a directed edge $\overrightarrow{(v_i,v_j)}$ where $1\le i< j\le \mathcal{M}$ if and only if $e_j+\delta >e_i$, and $\mathcal{M}_i,\mathcal{M}_j$ are compatible, i.e., form a common subsequence for $A$ and $B$. 
Figure~\ref{fig:ex}(c) illustrates an  example for $G$ where there is no edge from $v_1$ to $v_2$. This is because despite $e_2+\delta >e_1$, the matching pairs $\mathcal{M}_1,\mathcal{M}_2$ do not form a common subsequence for $A$ and $B$.  We maintain the points $(k,i)$, where $1\le k\le m$ and $i\in match(k)$, in a 2D range tree data structure so that given $v_i$, its outgoing edges can be found in $O(\log \mathcal{M} + c_i)$ time, where $c_i$ is the number of pairs that are compatible with $\mathcal{M}_i$. This is done by searching (using the range tree data structure) for the points with both coordinates higher than the point represented by $v_i$ and then removing the candidates that do not meet the almost increasing condition. In Figure~\ref{fig:ex}(b), there are three pairs compatible with vertex $v_1$ but only two of them satisfy the almost increasing condition. Hence $v_1$ has two outgoing edges. A range tree data structure can be computed in $O(\mathcal{M}\log \mathcal{M})$ time and space~\cite{de2000computational}. 

 We maintain a list of pairs at each $v_i$, where $1\le i\le \mathcal{M}$. A tuple $(r,w,v)$ at $v_i$ indicates that there exists an LCaIS of length $r$ that ends with the matching pair $\mathcal{M}_i$, where the largest element used in the LCaIS is $w$  and the matching pair just before $v_i$ is $v$. During our algorithm, we will minimize $w$ over all LCaIS of length $r$  that ends with   matching pair $\mathcal{M}_i$.

Initially, each $v_i$ contains the tuple $(1,e_i,-1)$ indicating that there exists an LCaIS of length 1 that ends with matching pair $\mathcal{M}_i$, where the largest element used is $e_i$ and the matching pair before $\mathcal{M}_i$ is null. We now iterate an index $i$  from 1 to $m$ and for each $v_i$, we consider the directed edges $\overrightarrow{(v_i,v_j)}$ with $i<j\le m$. When processing an edge $\overrightarrow{(v_i,v_j)}$ we check whether the tuples at $v_j$ can be updated. Let $(r_i,w_i,v)$ be a tuple at $v_i$. We now consider the following cases.
\begin{itemize}
    
    \item If $e_j+\delta > w_i$, then we have a potential for updating the tuples at $v_j$. 
    
    \begin{itemize}
        \item If $v_j$ does not contain any tuple that starts with $1+r_i$, then we can insert the tuple $(1+r_i,\max(w_i,e_j),v_i)$. This is because we now can form an LCaIS of length $1+r_i$ that ends with $\mathcal{M}_j$. Figure~\ref{fig:ex}(d) shows two such insertions in red. 
        \item If $v_j$ already contains a tuple $(1+r_i, w_j,v')$, then we check whether the condition  $\max(w_i,e_j)<w_j$ holds. 
        If so, then we update the tuple with $(1+r_i,\max(w_i,e_j),v_i)$. This is because we found  a better LCaIS of length $1+r_i$, i.e., the maximum element of the LCaIS computed by appending $\mathcal{M}_j$ has a lower value. 
        Figure~\ref{fig:ex}(e)  shows such an update in blue.
    \end{itemize}

    \item If $e_j+\delta \le w_i$, then we do not update any tuple at $v_j$. This is because $\mathcal{M}_j$ cannot be added to   form a longer LCaIS.

\end{itemize}

\begin{theorem}
    Given two number sequences $A[1\ldots n]$ and $B[1\ldots m]$ where $m\le n$, and a positive constant $\delta$, an LCaIS for these sequences can be computed in $O((n+m)\log n +\mathcal{M}\log \mathcal{M} + \mathcal{C}\ell)$ time and $O(\mathcal{M}(\ell+\log \mathcal{M}))$ space, where $\ell$ is the length of the LCaIS,    $\mathcal{M}$ is the number of matching pairs, and $\mathcal{C}$ is the number of compatible matching pairs. 
\end{theorem}
\begin{proof}
We now show that at the end of processing all the edges,  there must be a tuple corresponding to an LCaIS. Furthermore, we can construct the LCaIS by repeatedly following the parent. To this end, it suffices to prove the following.

\begin{enumerate}
    \item[] \emph{Claim: By the time we start processing $v_i$, where $1\le i\le \mathcal{M}$, it already stores a tuple $(r,w,v)$ at $v_i$ if and only if there exists an LCaIS of length $r$ with largest element $w$ that ends with $\mathcal{M}_i$. Here $w$ is the minimum possible, i.e., there does not exist another LCaIS of length $r$ that ends with $\mathcal{M}_i$ with largest element smaller than $w$. Furthermore, we can construct such an  LCaIS by taking $e_i$ and repeatedly following the parent starting at $v$.}

\smallskip
The case when $i=1$ is straightforward because $v_1$ only contains the tuple $(1,e_1,-1)$ and no LCaIS of length larger than 1 that ends at $\mathcal{M}_i$ is possible.  Assume now that the claim holds for any $1\le k < i$, and we have reached $v_i$. This means all the edges $\overrightarrow{(v_k,v_i)}$ have been processed already. 

\smallskip
If no such edge $\overrightarrow{(v_k,v_i)}$ exists then $v_i$ contains   the tuple $(1,e_i,-1)$ and no LCaIS of length larger than 1 that ends at $\mathcal{M}_i$ is possible. Otherwise, assume an edge $\overrightarrow{(v_k,v_i)}$, where $v_k $ has a tuple $(r_k,w_k,v)$. 

\smallskip
Since $k<i$, by induction hypothesis there exists an LCaIS of length $r_k$ that ends with $\mathcal{M}_k$ and $w_k$ is the smallest integer such that there exists an LCaIS of length $r_k$ that ends with $\mathcal{M}_k$. Furthermore, we can construct such an LCaIS by taking $\mathcal{M}_k$ and repeatedly following the parent starting at $v$. If $e_i+\delta > w_k$, then we either insert or update an existing tuple at $v_i$. Therefore, after processing all the edges $\overrightarrow{(v_k,v_i)}$, $v_i$ contains  a tuple $(r_i,w_i,v_x)$ if and only if there exists an LCaIS of length $r_i$ with last element $\mathcal{M}_i$ where $w_i$ minimizes the largest element over all such LCaIS. By induction hypothesis, $v_x$ contains a tuple $(r_i-1, w_x,v')$. Since $e_i+\delta >w_x$, we can find an LCaIS of length $r_i$ by taking $\mathcal{M}_i$ and repeatedly following the parent links.  
\end{enumerate}

\noindent
We now analyze the time and space complexity. We do not need to create the graph explicitly. For each $v_i$ we first find its compatible matching pairs in $O(\log \mathcal{M} + c_i)$ time and then determine the outgoing edges of $v_i$ within the same running time. Hence we need $O(\sum_{1\le i\le \mathcal{M} } (\log \mathcal{M} + c_i)) = O(\mathcal{M}\log \mathcal{M} + \mathcal{C})$ time to determine the edges of the graph. 
Each edge $\overrightarrow{(v_i,v_j)}$ is processed exactly once and  updating the tuples at $v_j$ takes $O(\ell)$ time, where $\ell$ is the length of the LCaIS of $A$ and $B$. Therefore, the total time to process all edges is $O(\mathcal{C}\ell)$. It is straightforward to construct the LCaIS in $O(\ell)$ time by following the parent links. Therefore, together with the time for determining matching pairs, computing the range tree data structure, and determining the graph edges, the overall running time becomes $O((n+m)\log n +\mathcal{M}\log \mathcal{M} + \mathcal{C}\ell)$.


The space needed corresponds to the range tree data structure plus   
the total number of tuples. The space for the former is $O(\mathcal{M}\log \mathcal{M})$, while the latter requires $O(\mathcal{M}\ell)$ space. Hence the  space complexity is $O(\mathcal{M}(\ell+\log \mathcal{M}))$. 
\end{proof}

\section{Conclusion}
\label{sec:con}
We gave  an  algorithm that given  two sequences of lengths $n$ and $m$ with $m\le n$, finds an LCaIS in $O(nm\ell)$-time and $O(n+m\ell)$-space, where $\ell$ is the length of the LCaIS. A natural avenue of research would be to improve the running time. Can we find some logarithmic factor improvement employing the techniques  used in designing  subquadratic-time algortihms for LCIS~\cite{agrawal2020faster,duraj2020sub}? It would also be interesting to examine the case when $\ell = o(\log n)$ to obtain faster algorithms by potentially leveraging bit-vector operations. We gave another 
$O((n+m)\log n +\mathcal{M}\log \mathcal{M} + \mathcal{C}\ell)$-time and $O(\mathcal{M}(\ell+\log \mathcal{M}))$-space
algorithm  which is faster when the number of matching pairs $\mathcal{M}$ and the number of compatible matching pairs $\mathcal{C}$ are in $o(nm/\log m)$. One may attempt to further improve the time and space complexities under these conditions. 

\subsection*{Acknowledgment} The research is supported in part by the Natural Sciences and Engineering Research Council of Canada (NSERC). We thank M. Sohel Rahman for our stimulating initial discussions and the anonymous reviewers for their feedback, which improved the presentation of the paper.
\bibliographystyle{splncs04}
\bibliography{sample}

\begin{thebibliography}{10}
\providecommand{\url}[1]{\texttt{#1}}
\providecommand{\urlprefix}{URL }
\providecommand{\doi}[1]{https://doi.org/#1}

\bibitem{abboud2015tight}
Abboud, A., Backurs, A., Williams, V.V.: Tight hardness results for {LCS} and other sequence similarity measures. In: 2015 IEEE 56th Annual Symposium on Foundations of Computer Science. pp. 59--78. IEEE (2015)

\bibitem{agrawal2020faster}
Agrawal, A., Gawrychowski, P.: A faster subquadratic algorithm for the longest common increasing subsequence problem. In: 31st International Symposium on Algorithms and Computation (ISAAC 2020). pp.~4--1. Schloss Dagstuhl--Leibniz-Zentrum f{\"u}r Informatik (2020)

\bibitem{bhuiyan2022computing}
Bhuiyan, M.T.H., Alam, M.R., Rahman, M.S.: Computing the longest common almost-increasing subsequence. Theoretical Computer Science  \textbf{930},  157--178 (2022)

\bibitem{bringman2018multivariate}
Bringman, K., K{\"u}nnemann, M.: Multivariate fine-grained complexity of longest common subsequence. In: Proceedings of the Twenty-Ninth Annual ACM-SIAM Symposium on Discrete Algorithms. pp. 1216--1235. SIAM (2018)

\bibitem{de2000computational}
De~Berg, M.: Computational geometry: algorithms and applications. Springer Science \& Business Media (2000)

\bibitem{duraj2020sub}
Duraj, L.: A sub-quadratic algorithm for the longest common increasing subsequence problem. In: 37th International Symposium on Theoretical Aspects of Computer Science (STACS 2020). pp. 41--1. Schloss Dagstuhl--Leibniz-Zentrum f{\"u}r Informatik (2020)

\bibitem{elmasry2010longest}
Elmasry, A.: The longest almost-increasing subsequence. In: International Computing and Combinatorics Conference. pp. 338--347. Springer (2010)

\bibitem{hirschberg1975linear}
Hirschberg, D.S.: A linear space algorithm for computing maximal common subsequences. Communications of the ACM  \textbf{18}(6),  341--343 (1975)

\bibitem{kutz2011faster}
Kutz, M., Brodal, G.S., Kaligosi, K., Katriel, I.: Faster algorithms for computing longest common increasing subsequences. Journal of Discrete Algorithms  \textbf{9}(4),  314--325 (2011)

\bibitem{moosa2013computing}
Moosa, J.M., Rahman, M.S., Zohora, F.T.: Computing a longest common subsequence that is almost increasing on sequences having no repeated elements. Journal of Discrete Algorithms  \textbf{20},  12--20 (2013)

\bibitem{schensted1961longest}
Schensted, C.: Longest increasing and decreasing subsequences. Canadian Journal of mathematics  \textbf{13},  179--191 (1961)

\bibitem{ta2021computing}
Ta, T.T., Shieh, Y.K., Lu, C.L.: Computing a longest common almost-increasing subsequence of two sequences. Theoretical Computer Science  \textbf{854},  44--51 (2021)

\bibitem{yang2005fast}
Yang, I.H., Huang, C.P., Chao, K.M.: A fast algorithm for computing a longest common increasing subsequence. Information Processing Letters  \textbf{93}(5),  249--253 (2005)

\end{thebibliography}




\end{document}